\documentclass{article}
\usepackage[utf8x]{inputenc}
\usepackage[T1]{fontenc}
\usepackage{makeidx}  % allows for indexgeneration
\usepackage{amsmath}
\usepackage{amssymb}
\usepackage{graphicx}
\usepackage{nicefrac}
\usepackage{subfigure}
\usepackage{todonotes}	% for notes
\usepackage{amsthm}
\theoremstyle{plain}
\newtheorem{theorem}{Theorem}

\newtheorem{corollary}[theorem]{Corollary}
\newtheorem{proposition}[theorem]{Proposition}
\begin{document}
\title{Congestion Games with Complementarities\thanks{This work was partially supported by the German Research Foundation (DFG) within the Collaborative Research Centre ``On-The-Fly Computing'' (SFB 901).\newline\newline The final publication is available at Springer via http://dx.doi.org/10.1007/978-3-319-57586-5\_19.}}
%
%\titlerunning{Congestion Games with Complementarities}  % abbreviated title (for running head)
\author{Matthias~Feldotto \and Lennart~Leder \and Alexander~Skopalik}
%
%\authorrunning{Feldotto et al.} % abbreviated author list (for running head)
%Craig Chambers, Kim B. Bruce, and Elisa Bertino}
%
%\institute{Heinz Nixdorf Institute \& Department of Computer Science \linebreak Paderborn University, Germany\\
%\email{\{feldi,lleder,skopalik\}@mail.upb.de}
%}

\maketitle              % typeset the title of the contribution

\begin{abstract} 
We study a model of selfish resource allocation that seeks to incorporate dependencies among resources as they exist in modern networked environments. Our model is inspired by utility functions with constant elasticity of substitution (CES) which is a well-studied model in economics. We consider congestion games with different aggregation functions. In particular, we study $L_p$ norms and analyze the existence and complexity of (approximate) pure Nash equilibria. Additionally, we give an almost tight characterization based on monotonicity properties to describe the set of aggregation functions that guarantee the existence of pure Nash equilibria.

%\keywords{Congestion Games, Aggregation, $L_p$ Norms, Complementarities, Existence of Equilibria, Approximate Pure Nash Equilibria}
\end{abstract}

\section{Introduction}

Modern networked environments often lack a central authority that has the ability or the necessary information to 
coordinate the allocation of resources such as bandwidths of network links, server capacities, cloud computing resources, etc. 
Hence, allocation decisions are delegated to local entities or customers.
Often they are interested in allocations that optimize for themselves rather than for overall system performance. 
We study the strategic interaction that arises in such situations using game theoretic methods. 

The class of congestion games~\cite{RO73} is a well-known model to study scenarios in which the players allocate shared resources.
In a congestion game each player chooses a subset of resources from a collection of allowed subsets which are called strategies. These resources may represent links in a network, servers, switches, etc.
Each resource is equipped with a cost function that is mapping from the number of players using it to a cost value.
The cost of a player is  the sum of the costs of the resources  in the chosen strategy. There are several well-known extensions to this model. In \emph{weighted} congestion games \cite{FO05}, players can have different weights and the cost of a resource depends on the total weight of the players using it. In \emph{player-specific} congestion games \cite{MI96}, the costs of resources can be different for different players.

However, all these models have in common that the cost of a player is defined as the sum of the resource costs. This is well-suited to describe latencies or delays of computer or traffic networks, for example.
However,  in scenarios in which bandwidth determines the costs of players this is determined by the bottleneck link.  To that end, bottleneck congestion games have been introduced \cite{BA07} in which the cost of a player is defined as the cost of her most expensive resource.
  
Both models are limited in their ability to model complementarities that naturally arise in scenarios where the performance of a resource depends to some degree on the performance of other resources. For example, a cloud-based web application may be comprised of many resources. A low performing resource negatively influences the performance of other parts and hence the overall system.
Bottleneck games assume perfect complements, whereas standard congestion games assume independence.
We seek to generalize both models and allow for different degrees of complementarity that may even differ between players. We are inspired by utility functions with constant elasticity of substitution (CES)~\cite{ACMS61,DS77}, which are a well-studied and accepted model in economics. We adapt the notion to our needs and study the analogue version for cost functions that corresponds to $L_p$ norms.
Clearly, both standard congestion games  and bottleneck congestion games are special cases of these games with $L_1$ and $L_\infty$ norms, respectively.
Using further aggregation functions instead of $L_p$ norms even allows to model more complex dependencies. Based on natural monotonicity properties of these functions, we can characterize the existence of pure Nash equilibria.

\subsection{Related Work}

Congestion games were introduced by Rosenthal~\cite{RO73} who shows that these games are potential games. In fact,  the class is isomorphic to the class of potential games as shown by Monderer and Shapley~\cite{MS96}.
The price of anarchy in the context of network congestion games was first considered by Koutsopias and Papadimitriou~\cite{KO99}. The related concept of smoothness, which can be used to derive a bound on the price of anarchy, was introduced by Roughgarden~\cite{RU09}.
Fabrikant et al.~\cite{FA04} show that in congestion games improvement sequences may have exponential length, and that it is in general PLS-complete to compute a pure Nash equilibrium.
Chien and Sinclair~\cite{CS07} show that for symmetric congestion games with a mild assumption on the cost function the approximate best-response dynamics converge quickly to an approximate pure Nash equilibrium. In contrast to that, Skopalik and Vöcking~\cite{SK08} show that it is in general even PLS-hard to compute approximate pure Nash equilibria for any polynomially computable approximation factor.
However, if the cost functions are restricted to linear or constant degree polynomials,
approximate pure Nash equilibria can be computed in polynomial time as shown by Caragiannis et al.~\cite{CA11},
even for weighted games~\cite{CA12} and some other variants~\cite{CFG14,FG14}.
Hansknecht et al.~\cite{HK14} use the concept of approximate potential functions to examine approximate pure Nash equilibria in weighted congestion games under different restrictions on the cost functions. For polynomial cost functions of maximal degree $g$ they show that $(g+1)$-approximate equilibria are guaranteed to exist.

Singleton congestion games, a class of congestion games which guarantees polynomial convergence of best-response improvement sequences to pure Nash equilibria, are considered by Ieong et al.~\cite{IE05}. They show that the property of polynomial convergence can be generalized to so called independent-resource congestion games. It is further generalized by Ackermann et al.~\cite{AC08} to matroid congestion games. They show that, for non-decreasing cost functions, the matroid property is not only sufficient, but also necessary to guarantee the convergence to pure Nash equilibria in polynomial time.
Milchtaich~\cite{MI96} studies the concept of player-specific congestion games and shows that in the singleton case these games always admit pure Nash equilibria. Ackermann et al.~\cite{AR09} generalize these results to matroid strategy spaces and show that the result also holds for weighted congestion games. They also examine the question of efficient computability and convergence towards these equilibria. Furthermore, they point out that in a natural sense the matroid property is maximal for the guaranteed existence of pure Nash equilibria in player-specific and weighted congestion games.
Moreover, Milchtaich~\cite{MI96} examines congestion games in which players are both weighted and have player-specific cost functions. By constructing a game with three players he shows that these games, even in the case of singleton strategies, do not necessarily possess pure Nash equilibria.
Mavronicalas et al.~\cite{MA07} study a special case of these games in which cost functions are not entirely player-specific. Instead, the player-specific resource costs are derived by combining the general resource cost function and a player-specific constant via a specified operation (e.g. addition or multiplication). They show that this restriction is sufficient to guarantee the existence of pure Nash-equilibria in games with three players.
Dunkel and Schulz~\cite{DS06} show that the decision problem whether a weighted network congestion game possesses a pure Nash equilibrium is NP-hard. For player-specific network congestion games, Ackermann and Skopalik~\cite{AS08} show that this problem is NP-complete both in directed and in undirected graphs.

Banner and Orda~\cite{BA07} introduce the class of bottleneck congestion games and study their applicability in network routing scenarios. In particular, they derive bounds on the price of anarchy in network bottleneck congestion games with restricted cost functions and show that there always exists a pure Nash equilibrium which is socially optimal, but that the computation of this equilibrium is NP-hard.
Harks et al.~\cite{HH13} give an overview on bottleneck congestion games and the complexity of computing pure Nash equilibria. Moreover, they show that in matroid bottleneck congestion games even pure strong equilibria, which are stable against coalitional deviations, can be computed efficiently.
Harks et al.~\cite{HA09} introduce the so called Lexicographical Improvement Property, which guarantees the existence of pure Nash equilibria through a potential function argument. They show that bottleneck congestion games fulfill this property.

Feldotto et al.~\cite{FLS16} generalize both variants and investigate the linear combination of standard and bottleneck congestion games. Kukushkin~\cite{KU07} introduces the concept of generalized congestion games in which players may use arbitrary monotonic aggregation functions to calculate their total cost from the costs of their single resources. He shows that, apart from monotonic mappings, additive aggregation functions that fulfill certain restrictions are the only ones for which the existence of PNE can be guaranteed. In a later paper~\cite{KU14}, he elaborates this result by deriving properties for the players' aggregation functions which are sufficient to establish this guarantee.
Another generalization of congestion games is given by Byde et al.~\cite{BPJ09} and Voice at al.~\cite{VPBJ09}. They introduce the model of games with congestion-averse utility functions. They show under which properties pure Nash equilibria exist and give a polynomial time algorithm to compute them.

\subsection{Our Contribution}

We introduce congestion games with $L_p$-aggregation functions 
 and  show that pure Nash equilibria are guaranteed only if either 
 there is one aggregation value $p$ for all players or in the case of matroid congestion games.
 For games with linear cost functions in which a pure Nash equilibrium exists, we derive bounds on the price of anarchy.
 For general games, we show the existence of approximate pure Nash equilibria where the approximation factor scales sublinearly with the size of the largest strategy set. We show that this factor is tight and that it is NP-hard to decide whether there is an approximate equilibrium with a smaller factor. Computing an approximate PNE with that factor is PLS-hard.
 As a positive result, we present two different polynomial time algorithms to compute approximate equilibria in games with linear cost functions. The approximation factors of both methods have a different dependence on the parameters of the game.
 For matroid games, we show the existence of pure Nash equilibria not only for $L_p$-aggregation functions but also seek to extend it to more general aggregation functions. We can characterize the functions that guarantee existence of PNE by certain monotonicity properties.

\subsection{Model/Preliminaries}

A \emph{congestion game with $L_p$-aggregation functions} is a tuple $\Gamma=\left(N,R,\left(\Sigma_i\right)_{i\in N},\right.$ $\left.\left(c_r\right)_{r\in R},(p_i)_{i\in N}\right)$. $N=\left\{1,\dots, n\right\}$ denotes the set of players, $R=\left\{r_1,\dots, r_m\right\}$ the set of resources. For each player $i\in N$,  $\Sigma_i \subseteq 2^R$ denotes the strategy space of player $i$ and $p_i\in \mathbb{R}$, $p_i\geq 1$ denotes the player-specific aggregation value of player $i$. For each resource $r$, $c_r : N \to \mathbb{R}$ denotes the non-decreasing cost function associated to resource $r$.

In a congestion game, the state $S=\left (S_1,\dots, S_n\right )$ describes the situation that each player $i\in N$ has chosen the strategy $S_i\in \Sigma_i$. In state $S$, we define for each resource $r\in R$ by
$n_r(S)=\left|\left \{i\in N~|~r\in S_i\right \}\right|$ the congestion of $r$.	The cost of resource $r$ in state $S$ is defined as $c_r(S)=c_r\left(n_r(S)\right)$. The cost of player $i$ is defined as 
$c_i(S)= \left(\sum_{r\in S_i} c_r(S)^{p_i}\right)^{\frac{1}{p_i}}.$
If for all $i,j\in N$ it holds that $p_i=p_j$, then we call $\Gamma$ a congestion game with identical $L_p$-aggregation functions. 

For a state $S=(S_1,...,S_i,..., S_n)$, we denote by $\left(S_{i}',S_{-i}\right)$ the state that is reached if player $i$ plays strategy $S_i'$ while all other strategies remain unchanged. 
A state $S=\left (S_1,\dots, S_n\right)$ is called a {\em pure Nash equilibrium (PNE)} if for all $i\in N$ and all $S'_i\in \Sigma_i$ it holds that $c_i(S)\leq  c_i(S'_i,S_{-i})$ and a {\em $\beta$-approximate pure Nash equilibrium} for a $\beta \ge 1$ if for all $i\in N$ and all $S'_i\in \Sigma_i$ it holds that $c_i(S)\leq \beta\cdot c_i(S'_i,S_{-i})$.
Additionally, we define the Price of Anarchy as the worst-case ratio between the costs in any equilibria and the minimal possible costs in the game. Formally, it is given by  $\frac{\max_{S \in \mathcal{PNE}}\sum_{i\in N}c_i(S)}{\min_{S \in \mathcal{S}}\sum_{i\in N}c_i(S^*)}$. A game is called $(\lambda, \mu)$-smooth for $\lambda >0$ and $\mu \leq 1$ if, for every pair of states $S$ and $S'$, we have $\sum_{i \in N}c_i(S'_i, S_{-i}) \leq \lambda \sum_{i\in N}c_i(S') + \mu \sum_{i\in N}c_i(S)$. In a $(\lambda, \mu)$-smooth game, the Price of Anarchy is at most $\frac{\lambda}{1-\mu}$~\cite{RU09}.

\section{Existence and Efficiency of Pure Nash Equilibria}

We begin with an easy observation that congestion games in which all players use the same $L_p$-norm as aggregation function always possesses a PNE.
\begin{proposition}
	\label{p-equal}
	Let $\Gamma$ be a congestion game with identical $L_p$-aggregation functions. Then $\Gamma$ possesses at least one pure Nash equilibrium.
\end{proposition}  

However, if players are heterogeneous in the sense that they use different aggregation functions a PNE might not exist even for two player games.

\begin{theorem}
\label{t-notexist}
	For every $1\leq p_1<p_2$ there exists a 2-player congestion game with $L_p$-aggregation functions $\Gamma$ that does not possess a pure Nash equilibrium.
\end{theorem} 

Now we will investigate the efficiency of the equilibria  by analyzing the price of anarchy. We restrict ourselves to games with linear cost functions and make use  of a previous result by Christodolou et al.~\cite{CH11}.
Let $q:=max_{i\in N}\ p_i$, let $d:=max_{i\in N,S_i\in\Sigma_i}|S_i|$.
Furthermore, let  $z=\left\lfloor \frac{1}{2}\cdot \left(d^{1-\frac{1}{q}}+\sqrt{5+6\cdot\left(d^{1-\frac{1}{q}}-1\right)}\right.\right.$ $\left.\left.\overline{+\left(d^{1-\frac{1}{q}}-1\right)^2}\right)\right\rfloor$ be the maximum integer such that $\frac{z^2}{z+1}\leq d^{1-\frac{1}{q}}$.

\begin{theorem}
	\label{PoA}
	A congestion game with $L_p$-aggregation functions $\Gamma$ in which all cost functions are linear  is $\left(d^{1-\frac{1}{p}}\cdot\frac{z^2+3z+1}{2z+1},d^{1-\frac{1}{p}}\cdot\frac{1}{2z+1}\right)$-smooth.
	If $\Gamma$ possesses a pure Nash equilibrium, its price of anarchy is bounded by 
	$d^{1-\frac{1}{p}}\cdot\frac{z^2+3z+1}{2z+1-d^{1-\frac{1}{p}}}$.	
\end{theorem}

\section{Existence of Approximate Pure Nash Equilibria}

We start by giving a bound depending on the minimal and maximal $p_i$-values that players use and the maximal number of resources in a strategy.

\begin{theorem}
	\label{Lp-approx}
	Let $\Gamma$ be a congestion game with $L_p$-aggregation functions, let $p=\min_{i\in N}p_i$ be the minimal and $q=\max_{i\in N}p_i$ be the maximal aggregation value in the game. Furthermore, denote by $d=max_{i\in N,S_i\in\Sigma_i}|S_i|$ the size of the strategy that contains most resources. Then $\Gamma$ contains a $\beta$-approximate equilibrium for $\beta=d^{\frac{1}{2}\cdot\left(\frac{1}{p}-\frac{1}{q}\right)}$. Moreover, a $\beta$-approximate equilibrium will be reached from an arbitrary state after a finite number of $\beta$-improvement steps.
\end{theorem}
In the proof 	we show that $\Phi(S)=\sum_{r\in R}\sum_{i=1}^{n_r(S)}c_r(i)^z$ is an approximate potential function where  $z:=\left(\frac{1}{2}\cdot \left(\frac{1}{p}+\frac{1}{q}\right)\right)^{-1}$.
We will now complement this result by showing that this approximation quality is the best achievable: i.e., we show that for any given $p<q$ and $d\geq 2$ we can construct a game
which possesses no $\beta$-approximate PNE for any $\beta<d^{\frac{1}{2}\left(\frac{1}{p}-\frac{1}{q}\right)}$.
\begin{theorem}
	\label{Lpnonapprox}
	Let $p, q, d\in \mathbb{N}$ with $p<q$ and $d\geq 2$. Then there is a congestion game with $L_p$-aggregation functions $\Gamma$ with $N=\{1,2\}$, $p_1=p$, $p_2=q$, and $d=max_{i\in N,S_i\in\Sigma_i}|S_i|$ such that $\Gamma$ does not possess a $\beta$-approximate pure Nash equilibrium for any $\beta<d^{\frac{1}{2}\left(\frac{1}{p}-\frac{1}{q}\right)}$.
\end{theorem}

We complete the discussion of approximate pure Nash equilibria by regarding the computational complexity of deciding whether an approximate PNE exists for any approximation factor smaller than $\beta$.

\begin{theorem}
	\label{Lpapproxcomp}
	For any $p,q,d\in \mathbb{N}$ with $p<q$ and $d\geq 2$ it is NP-hard to decide whether a given congestion game with $L_p$-aggregation functions $\Gamma$, with $p\leq p_i\leq q$ for all $i\in N$, and $d=\max_{i\in N, S_i\in \Sigma_i} |S_i|$ possesses a $\beta$-approximate pure Nash equilibrium for any $\beta<d^{\frac{1}{2}\cdot\left(\frac{1}{p}-\frac{1}{q}\right)}$.
\end{theorem}

\section{Computation of Approximate Equilibria} 

In~\cite{SK08} it was shown that it is PLS-hard to compute an $\beta$-approximate PNE in standard congestion games. Since these games are a special case of congestion games with $L_p$-aggregation functions, this negative result immediately carries over to congestion games with $L_p$-aggregation functions.
\begin{proposition}
	It is PLS-hard to compute a $\beta$-approximate pure Nash equilibrium in a congestion game $\Gamma$ with $L_p$-aggregation functions in which all cost functions are non-negative and non-decreasing, for any $\beta$ that is computable in polynomial time. 
\end{proposition}

In the light of this initial negative result, we consider games with restricted cost functions. Caragiannis et al.~\cite{CA11} provide an algorithm that computes approximate pure Nash equilibria for congestion games with polynomial cost functions. For linear costs, the algorithm achieves an approximation quality of $2+\epsilon$. For  polynomial functions with a maximal degree of $g$, the algorithm guarantees an approximation factor of $g^{O(g)}$.
We will reuse the algorithmic idea in two different ways which yield to different approximation guarantees depending on the aggregation parameters $p_i$.

\begin{theorem}
	Let $\Gamma$ be a congestion game with $L_p$-aggregation functions in which all cost functions are linear or polynomial functions of degree at most $g$ without negative coefficients. Furthermore, let $p:=min_{i\in N}\ p_i$, let $q:=max_{i\in N}\ p_i$, let $d:=max_{i\in N,S_i\in\Sigma_i}|S_i|$, and $z=\left(\frac{1}{2}\cdot \left(\frac{1}{p}+\frac{1}{q}\right)\right)^{-1}$.\\
	Then an $\beta$-approximate equilibrium of $\Gamma$ can be computed in polynomial time for $\beta= \min{\left\{(2+\epsilon)\cdot d^{1-\frac{1}{q}},z^{O(1)}\cdot d^{\frac{1}{2}\left(\frac{1}{p}-\frac{1}{q}\right)}\right\}}$ (linear cost functions) and for $\beta=\min{\left\{g^{O(g)}\cdot d^{1-\frac{1}{q}}, (g\cdot z)^{O(g)}\cdot d^{\frac{1}{2}\left(\frac{1}{p}-\frac{1}{q}\right)}\right\}}$ (polynomial cost functions).
\end{theorem}

\begin{proof}
	We apply the algorithm proposed by Caragiannis et al.~\cite{CA11} to $\Gamma$, disregarding the aggregation values. Since all cost functions are either linear or polynomial, the algorithm computes either a $(2+\epsilon)$- or a $g^{O(g)}$-approximate equilibrium.
	Now we can use the proof of Theorem \ref{Lp-approx} (for $z=1$). We get that in the state computed by the algorithm, which would be a either $(2+\epsilon)$- or $g^{O(g)}$-approximate PNE if all players used the $L_1$-norm, no player $i$ can improve her costs according to the $L_{p_i}$-norm by more than a factor of either $(2+\epsilon)\cdot d^{1-\frac{1}{p_i}}\leq\beta$ or $g^{O(g)}\cdot d^{1-\frac{1}{p_i}}\leq\beta$ . Hence, the computed state is an $\beta$-approximate pure Nash equilibrium in $\Gamma$. As analyzed in~\cite{CA11}, the running time of the algorithm is polynomial in the size of $\Gamma$ and $\frac{1}{\epsilon}$.
	For the second approximation factor and linear costs functions we replace every $c(x)$ in $\Gamma$ by a polynomial cost function $c'(x)=c(x)^z$ of degree $z$, where for simplicity $z$ is assumed to be integral. For this game, the algorithm given in~\cite{CA11} computes a state $S$ which is a $z^{O(z)}$-approximate equilibrium. The costs of all players are equal to the costs they would have in $\Gamma$ if they accumulated their costs according to the $L_z$-norm without taking the $z$-th root. Following the argumentation of the proof of Theorem \ref{Lp-approx}, we get for any player $i$ and any strategy $S_i'\in \Sigma_i$: 
	\begin{align*}
	\frac{c_i(S)}{c_i(S_i',S_{-i})} \leq \left(\left(z^{O(z)}\right)^\frac{p_i}{z}\right)^\frac{1}{p_i}\cdot\left(d^{\frac{p_i}{z}-1}\right)^\frac{1}{p_i}=z^{O(1)}\cdot d^{\frac{1}{z}-\frac{1}{p_i}}\leq z^{O(1)}\cdot d^{\frac{1}{2}\left(\frac{1}{p}-\frac{1}{q}\right)}.
	\end{align*}
	Obviously, the transformation of the cost functions can be done in polynomial time.
	Hence, the algorithm given in~\cite{CA11} computes a $z^{O(1)}\cdot d^{\frac{1}{2}\left(\frac{1}{p}-\frac{1}{q}\right)}$-approximate equilibrium of $\Gamma$ in polynomial time. 
	For polynomial cost functions this will lead to a game with polynomial cost functions of a degree of at most $g\cdot z$. Hence, the algorithm from~\cite{CA11} computes a $(g\cdot z)^{O(g\cdot z)}$-approximate PNE. Following the reasoning of the proof, we get that the computed state is a $(g\cdot z)^{O(g)}\cdot d^{\frac{1}{2}\left(\frac{1}{p}-\frac{1}{q}\right)}$-approximate PNE of the congestion game with $L_p$-aggregation functions.
\qed
\end{proof}

We have derived two different upper bounds for the approximation quality of approximate equilibria that can be computed in polynomial time. Generally speaking, if $d$ is small but players use high aggregation values, the first strategy yields the better approximation, while otherwise the second bound is better.

\section{General Aggregation Functions in Matroid\\ Games}
\label{aggregationsection} 

In this section we extend our model to a more general class of aggregation functions. Instead of using the $L_p$ norms in the cost functions of the player, they are now defined by $c_i(S)= f_i\left(c_{r_1}(S), c_{r_2}(S), \ldots, c_{r_m}\right)$ with $f_i$ being an arbitrary aggregation function for each player based on certain monotonicity properties. We now consider only matroid congestion games in which the strategy spaces of all players form the bases of a matroid on the set of resources.

Let $f:\mathbb{R}^d\to \mathbb{R}$ be a function that is defined on non-decreasingly ordered vectors.
Let for all $b=(b_1,\dots,b_d)$ and $b'=(b_1',\dots,b_d')$ with $b_i\leq b_i'$ for all $1\leq i\leq d$ hold that $f(b)\leq f(b')$. Then $f$ is called \textbf{strongly monotone}. 
Let $x=(x_1,\dots,x_d)$ and $y=(y_1,\dots,y_d)$ be vectors that differ in only one element, i.e., there are indices $j$ and $k$ such that $x_i=y_i$ for all $i<j$ and $i>k$, $x_j<y_k$, and $x_{i+1}=y_i$ for all $j\leq i< k$ and $f(y)<f(x)$. Furthermore, let there be a vector $z=(z_1,\dots,z_{d-1})$ such that $f(z_1,\dots,x_j,\dots,z_{d-1})<f(z_1,\dots,y_k,\dots,z_{d-1})$ (with $x_j$ and $y_k$ at their correct positions in the non-decreasingly ordered vectors). Then $f$ is called a \textbf{strongly non-monotone} function.
If $f$ is not strongly non-monotone, then $f$ is called a \textbf{weakly monotone} function.
We remark that this definition of vectors that differ in only one element does not require these elements to be at the same position in the vectors (the case $j=k$). It is sufficient that the symmetric difference of the multisets containing all elements in $x$ and $y$ contains exactly two elements $x_j$ and $y_k$.

\begin{theorem}
	\label{Mataggr}
	Let $\Gamma$ be a matroid congestion game in which each player has a personal cost aggregation function $f_i$ according to which her costs are calculated from her single resource costs. If for all players $i\in N$ the aggregation function $f_i$ is strongly monotone, then $\Gamma$ contains a pure Nash equilibrium. 
\end{theorem}  

\begin{proof}
	It is sufficient to show that any strategy $S_i=\{r_1,\dots,r_d\}$ that minimizes the sum $\sum_{j=1}^d c_{r_i}(S)$ in a state $S$ also minimizes the cost $f_i\left(c_i(S)\right)$, where $c_i(S)$ denotes the non-decreasingly ordered vector of resource costs of player $i$ in state $S$. Then a pure Nash equilibrium can be computed by computing a PNE in the corresponding game in which all players use the $L_1$-norm.
	We can show that if $B=\{b_1,\dots,b_d\}$ is a matroid basis that is minimal w.r.t. the sum $\sum_{i=1}^d b_i$, then for any other basis $B'=\{b'_1,\dots,b'_d\}$ and all $1\leq i \leq d$ it holds that $b_i\leq b'_i$ (w.l.o.g. assume that both $B$ and $B'$ are written in non-decreasing order).
	Hence, if $B$ is a basis that is optimal w.r.t. sum costs, then for all other bases $B'$ it holds that $f(B)\leq f(B')$, since $f$ is a strongly monotone function.
	\qed
\end{proof}
Since all $L_p$ norms are monotone functions, we can extend this result to congestion games with $L_p$ norms:

\begin{corollary}
	Let $\Gamma$ be a matroid congestion game with $L_p$ norms, then $\Gamma$ contains a pure Nash equilibrium. 
\end{corollary}

We have shown that strongly monotone aggregation functions are sufficient to guarantee the existence of a PNE in matroid congestion games with player-specific aggregation functions. This immediately gives rise to the question whe\-ther the monotonicity criterion is also necessary to achieve this guarantee. We investigate this question by examining if, given a non-monotone aggregation function $f$, we can construct a matroid congestion game in which all players use $f$ and which does not contain a PNE.
For singletons, we can immediately give a negative answer to this. Since in this case costs can be associated to single resources and all players use the same aggregation function, Rosenthal's potential function argument~\cite{RO73} is applicable and shows that a PNE necessarily exists.
However, for matroid degrees of at least $2$ the answer is positive if the aggregation function fulfills the property that we call strong non-monotonicity.

\begin{theorem}
	\label{strongnonmonotone}
	Let $f:\mathbb{R}^d\to \mathbb{R}$, $d\geq 2$ be a strongly non-monotone function. Then there is a 2-player matroid congestion game in which both players allocate matroids of degree $d$ and use the aggregation function $f$, which does not contain a pure Nash equilibrium.
\end{theorem}

As argued, it is reasonable to demand that the aggregation function $f$ in the proof is strongly non-monotone. We will underline this by showing that the strong non-monotonicity actually is a sharp criterion: i.e., it is both sufficient and necessary to construct a game without a PNE from $f$.

\begin{theorem}
	\label{weaklymonotone}
	Let $f$ be an aggregation function that is weakly monotone and let $\Gamma$ be a matroid congestion game in which all players use $f$ as their aggregation function. Then $\Gamma$ possesses a pure Nash equilibrium. Furthermore, from every state there is a sequence of best-response improvement steps that reaches a pure Nash equilibrium after a polynomial number of steps.
\end{theorem}
\begin{proof}
	Since $f$ is weakly monotone, we have for all vectors $v_x$ and $v_y$ which differ in exactly one component (let $v_x$ contain $x$ and $v_y$ contain $y$, with $x<y$) that either $f(v_x)\leq f(v_y)$ or $f(v_y)<f(v_x)$ and for any vector $w_y$ that contains $y$ it holds that $f(w_y)\leq f(w_x)$, where $w_x$ results from $w_y$ by replacing $y$ by $x$. Hence, for all pairs $(x,y)$ we have either $f(w_x)\leq f(w_y)$ for all $w_x$ and $w_y$, or $f(w_y)\leq f(w_x)$ for all $w_x$ and $w_y$. This means that if we replace one element by another one in an arbitrary vector, the direction in which the value of $f$ changes (if it changes at all) depends only on the two exchanged elements, not on the rest of the vector.
	Based on this, we define the relation $\leq'$ on the real numbers by determining that
	$x\leq' y$ if and only if $f(w_x)\leq f(w_y)$ for all vectors $w_x$ and $w_y$.
	As argued, this relation defines a total preorder on $\mathbb{R}$. Since the number of resources and players in the game $\Gamma$ are finite, the number of different resource costs that can occur in the game is also finite. Hence, it is possible to enumerate all possible resource costs according to the ordering relation $\leq'$. We denote the position of the cost value $c_r(S)$ in this enumeration by $\pi(c_r(S))$: i.e., $\pi(c_r(S))=1$ if and only if for all $r'\in R$ and all $l\in N$ it holds that $c_r(S)\leq' c_{r'}(l)$. We have that $\pi(c_r(S))=\pi(c_{r'}(S'))$ if and only if $c_r(S)\leq' c_{r'}(S')$ and $c_{r'}(S')\leq' c_{r}(S)$. We remark that $\leq'$ is not necessarily a total order. Thus the two cost values need not be equal in this case.

	Using this, we define the potential function 
	$\Phi(S)=\sum_{r\in R}\sum_{i=1}^{n_r(S)}\pi(c_r(S))$.
	Consider a state $S=(S_i,S_{-i})$ in which player $i$ can improve her cost by deviating to the strategy $S_i'$, yielding the state $S'=(S_i',S_{-i})$. Since $S_i$ and $S_i'$ are both bases of the same matroid $M$, the graph $G=(V,E)$ with $V=(S_i\setminus S_i' \cup S_i'\setminus S_i)$ and $E=\left\{\{r,r'\}~|~r\in S_i, r'\in S_i', S_i'\setminus \{r'\}\cup\{r\}\in M\right\}$ contains a perfect matching (see Corollary 39.12a in~\cite{SC03}).
	All edges in $G$ correspond to resource pairs $\{r,r'\}$ such that $S_i'\setminus \{r'\}\cup\{r\}$ is a valid strategy for player $i$. Since $S_i'$ is a best response strategy to the strategy profile $S_{-i}$ of all other players, it must hold for all edges $\{r,r'\}$ that $f(S_i'\setminus \{r'\}\cup\{r\},S_{-i})\geq f(S_i',S_{-i})$. This implies that either $c_{r'}(S')\leq' c_r(S)$ or $f(S_i'\setminus \{r'\}\cup\{r\},S_{-i})= f(S_i',S_{-i})$ and $c_{r'}(S')>c_r(S)$. In the latter case, the strategy $S_i'\setminus \{r'\}\cup\{r\}$ is still a best-response strategy for player $i$. Repeating the argument yields that there must be a best-response strategy $S_i''$ such that in the graph defined analogously to $G$ it holds for all edges $\{r,r'\}$ that $c_{r'}(S'')\leq' c_r(S)$, where $S''=(S_i'',S_{-i})$.
	
	Let $T=\{e_1,\dots,e_k\}$ be a perfect matching in this graph. 
	For all $\{r,r'\}\in T$ it holds that $c_{r'}(S'')\leq' c_r(S)$, and hence $\pi(c_{r'}(S''))\leq \pi(c_r(S))$.
	We have to argue that $T$ contains at least one edge $\{r,r'\}$ with $\pi(c_{r'}(S''))<\pi(c_r(S))$.
	Assume that for all $\{r,r'\}\in T$ it held that $\pi(c_{r'}(S'))=\pi(c_r(S))$, i.e., $c_r(S)\leq' c_{r'}(S'')$. Then we could transform $S_i$ into $S_i''$ by iteratively exchanging a single resource $r$ for another resource $r'$. Since two consecutive sets $S^r$ and $S^{r'}$ in this sequence differ only in the resources $r$ and $r'$, and $c_r(S)\leq' c_{r'}(S')$, it holds that $f(S^r)\leq f(S^{r'})$. Hence, none of the steps decreases the value of $f$, which contradicts the assumption that $f(S_i'',S_{-i})\leq f(S_i',S_{-i})<f(S_i,S_{-i})$. 
	Therefore, there must be at least one edge $\{r,r'\}$ in $T$ with $\pi(c_{r'}(S''))<\pi(c_r(S))$, which implies 
	$\Phi(S'')-\Phi(S)=\sum_{r'\in S_i''\setminus S_i}\pi(c_{r'}(S''))-\sum_{r\in S_i\setminus S_i''}\pi(c_{r}(S))
	=\sum_{\{r,r'\}\in T}\left(\pi(c_r'(S''))-\pi(c_r(S))\right)<0.$
	By construction, the value of $\Phi$ is always integral and upper bounded by $n^2\cdot m^2$, where $n$ is the number of players and $m$ the number of resources in $\Gamma$. Hence, $\Gamma$ reaches a PNE from an arbitrary state after at most $n^2\cdot m^2$ best-response improvement steps.
	\qed
\end{proof}

The theorem states that strong non-monotonicity is necessary to construct a game without a PNE from a \emph{single} aggregation function $f$. However, the technique used in the proof only requires that all aggregation functions used in the game have the same order on vectors which differ in exactly one component. It is irrelevant how these functions order vectors which differ in several components.
\begin{corollary}
	Let $\Gamma$ be a matroid congestion game in which the costs of player $i$ are computed according to her personal aggregation function $f_i$.  If for all $i\in N$ the function $f_i$ is weakly monotone and for all $i,j\in N$ and all vectors $v$ and $w$ that differ in exactly one component it holds that $f_i(v)\leq f_i(w)\Leftrightarrow f_j(v)\leq f_j(w)$, then $\Gamma$ possesses a pure Nash equilibrium. 
\end{corollary}
This corollary is interesting mainly because it establishes an almost tight border up to which the existence of PNE can be guaranteed. If we are given two aggregation functions $f$ and $g$ and two vectors $v$ and $w$ that differ in exactly one component, with $f(v)<f(w)$ and $g(w)<g(v)$, then it is obvious that we can construct a 2-player game in which the first player uses the aggregation function $f$ and the second $g$ and the two players alternate between the cost vectors $v$ and $w$, as we did in the proof of Theorem \ref{strongnonmonotone} for strongly non-monotone functions.

\section{Conclusion}

For congestion games with $L_p$-aggregation functions, we presented methods to compute approximate PNE and bound the price of anarchy which are based on previous results regarding standard congestion games. It is an open point for future work to examine if these results could be improved by specifically designing methods for congestion games with $L_p$-aggregation functions. Another interesting approach for further research would be to combine the application of aggregation functions with other classes of congestion games such as weighted congestion games or non-atomic congestion games, and examine the implications for the (approximate) pure Nash equilibria in these games.

\bibliographystyle{splncs03}
%\bibliography{references}

\begin{thebibliography}{11}
	
	\bibitem{AR09}
	Ackermann, H., Röglin, H., Vöcking, B.: Pure nash equilibria in
	player-specific and weighted congestion games. Theoretical Computer Science
	410(17),  1552--1563 (2009)
	
	\bibitem{AC08}
	Ackermann, H., R\"{o}glin, H., V\"{o}cking, B.: {On the Impact of Combinatorial
		Structure on Congestion Games}. Journal of the ACM  55(6),  25:1--25:22 (2008)
	
	\bibitem{AS08}
	Ackermann, H., Skopalik, A.: {Complexity of Pure Nash Equilibria in
		Player-Specific Network Congestion Games}. Internet Mathematics  5(4),
	323--342 (2008)
	
	\bibitem{ACMS61}
	Arrow, K.J., Chenery, H.B., Minhas, B.S., Solow, R.M.: Capital-labor
	substitution and economic efficiency. The Review of Economics and Statistics
	43(3),  225--250 (1961)
	
	\bibitem{BA07}
	Banner, R., Orda, A.: {Bottleneck Routing Games in Communication Networks}.
	IEEE Journal on Selected Areas in Communications  25(6),  1173--1179 (2007)
	
	\bibitem{BPJ09}
	Byde, A., Polukarov, M., Jennings, N.R.: Games with Congestion-Averse
	Utilities. In:
	%Mavronicolas, M., Papadopoulou, V.G. (eds.)
	Proceedings of the Second International Symposium on Algorithmic Game Theory
	(SAGT 2009). pp. 220--232. Springer Berlin Heidelberg (2009)
	
	\bibitem{CA11}
	Caragiannis, I., Fanelli, A., Gravin, N., Skopalik, A.: {Efficient Computation
		of Approximate Pure Nash Equilibria in Congestion Games}.
	In: IEEE 52nd Annual Symposium on Foundations of Computer
	Science (FOCS 2011). pp. 532--541 (2011)
	
	\bibitem{CFG14}
	Caragiannis, I., Fanelli, A., Gravin, N.: Short Sequences of Improvement Moves
	Lead to Approximate Equilibria in Constraint Satisfaction Games. In: Proceedings
	% Lavi, R. (ed.)
	of the 7th International Symposium on Algorithmic Game Theory (SAGT 2014).
	pp. 49--60. Springer Berlin Heidelberg (2014)
	
	\bibitem{CA12}
	Caragiannis, I., Fanelli, A., Gravin, N., Skopalik, A.: {Approximate Pure Nash
		Equilibria in Weighted Congestion Games: Existence, Efficient Computation,
		and Structure}. ACM Transactions on Economics and Computation  3(1),  2:1--2:32 (2015)
	
	\bibitem{CS07}
	Chien, S., Sinclair, A.: Convergence to Approximate Nash Equilibria in
	Congestion Games. Games and Economic Behavior  71(2),  315--327 (2011)
	
	\bibitem{CH11}
	Christodoulou, G., Koutsoupias, E., Spirakis, P.G.: On the Performance of
	Approximate Equilibria in Congestion Games. Algorithmica  61(1),  116--140
	(2011)
	
	\bibitem{DS77}
	Dixit, A.K., Stiglitz, J.E.: Monopolistic Competition and Optimum Product
	Diversity. The American Economic Review  67(3),  297--308 (1977)
	
	\bibitem{DS06}
	Dunkel, J., Schulz, A.S.: {On the Complexity of Pure-Strategy Nash Equilibria
		in Congestion and Local-Effect Games}. Mathematics of Operations Research
	33(4),  851--868 (2008)
	
	\bibitem{FA04}
	Fabrikant, A., Papadimitriou, C., Talwar, K.: {The Complexity of Pure Nash
		Equilibria}. In: Proceedings of the Thirty-sixth Annual ACM Symposium on
	Theory of Computing (STOC 2004). pp. 604--612., ACM, New York, NY, USA (2004)
	
	\bibitem{FG14}
	Feldotto, M., Gairing, M., Skopalik, A.: {Bounding the Potential Function in
		Congestion Games and Approximate Pure Nash Equilibria}. In: 
	%Liu, T., Qi, Q., Ye, Y. (eds.)
	Proceedings of the 10th International Conference on Web and Internet Economics
	(WINE 2014). LNCS, vol. 8877, pp. 30--43. Springer (2014)
	
	\bibitem{FLS16}
	Feldotto, M., Leder, L., Skopalik, A.: Congestion Games with Mixed Objectives.
	In: 
	%Chan, T.H.H., Li, M., Wang, L. (eds.)
	Proceedings of the 10th International Conference on Combinatorial Optimization and
	Applications (COCOA 2016). pp. 655--669. Springer, Cham (2016)
	
	\bibitem{FO05}
	Fotakis, D., Kontogiannis, S., Spirakis, P.: Selfish unsplittable flows.
	Theoretical Computer Science  348(2),  226--239 (2005)
	
	\bibitem{GJ02}
	Garey, M.R., Johnson, D.S.: {Computers and Intractability: {A} Guide to the
		Theory of NP-Completeness}. W. H. Freeman (1979)
	
	\bibitem{HK14}
	Hansknecht, C., Klimm, M., Skopalik, A.: {Approximate Pure Nash Equilibria in
		Weighted Congestion Games}. In:
	%Jansen, K., Rolim, J.D.P., Devanur, N.R., Moore, C. (eds.) 
	Approximation, Randomization, and Combinatorial
	Optimization. Algorithms and Techniques (APPROX/RANDOM 2014). LIPIcs, vol.~28,
	pp. 242--257. Dagstuhl, Germany (2014)
	
	\bibitem{HH13}
	Harks, T., Hoefer, M., Klimm, M., Skopalik, A.: Computing pure Nash and strong
	equilibria in bottleneck congestion games. Mathematical Programming  141(1),
	193--215 (2013)
	
	\bibitem{HA09}
	Harks, T., Klimm, M., M{\"{o}}hring, R.H.: {Strong Nash Equilibria in Games
		with the Lexicographical Improvement Property}. In: 
	%Leonardi, S. (ed.)
	Proceedings of the 5th International Workshop on Internet and Network Economics
	(WINE 2009). LNCS, vol. 5929, pp. 463--470. Springer (2009)
	
	\bibitem{IE05}
	Ieong, S., McGrew, R., Nudelman, E., Shoham, Y., Sun, Q.: Fast and Compact: {A}
	Simple Class Of Congestion Games. In: {AAAI}. pp. 489--494. {AAAI} Press /
	The {MIT} Press (2005)
	
	\bibitem{KL12}
	Klimm, M.: Competition for Resources: The Equilibrium Existence Problem in
	Congestion Games. Ph.D. thesis, Technische Universität Berlin (2012)
	
	\bibitem{KO99}
	Koutsoupias, E., Papadimitriou, C.: Worst-Case Equilibria. In:
	% Meinel, C., Tison, S. (eds.)
	Proceedings of the 16th Annual Symposium on Theoretical Aspects of
	Computer Science (STACS 1999). pp. 404--413.
	Springer Berlin Heidelberg (1999)
	
	\bibitem{KU14}
	Kukushkin, N.S.: Rosenthal's Potential and a Discrete Version of the
	Debreu--Gorman Theorem. Automation and Remote Control  76(6),  1101--1110
	(2015)
	
	\bibitem{KU07}
	Kukushkin, N.S.: Congestion games revisited. International Journal of Game
	Theory  36(1),  57--83 (2007)
	
	\bibitem{MA07}
	Mavronicolas, M., Milchtaich, I., Monien, B., Tiemann, K.: {Congestion Games
		with Player-Specific Constants}. In:
	% Kucera, L., Kucera, A. (eds.)
	Proceedings of the 32nd International Symposium on Mathematical Foundations
	of Computer Science (MFCS 2007). LNCS, vol. 4708, pp.
	633--644. Springer (2007)
	
	\bibitem{MI96}
	Milchtaich, I.: Congestion Games with Player-Specific Payoff Functions. Games
	and Economic Behavior  13(1),  111--124 (1996)
	
	\bibitem{MS96}
	Monderer, D., Shapley, L.S.: Potential Games. Games and Economic Behavior
	14(1),  124--143 (1996)
	
	\bibitem{RO73}
	Rosenthal, R.W.: A Class of Games Possessing Pure-Strategy Nash Equilibria.
	International Journal of Game Theory  2(1),  65--67 (1973)
	
	\bibitem{RU09}
	Roughgarden, T.: Intrinsic Robustness of the Price of Anarchy.
	Journal of the ACM  62(5), 32:1--32:42 (2015)
	
	\bibitem{SC03}
	Schrijver, A.: Combinatorial Optimization: Polyhedra and Efficiency, vol.~24.
	Springer Science \& Business Media (2002)
	
	\bibitem{SK08}
	Skopalik, A., V\"{o}cking, B.: {Inapproximability of Pure Nash Equilibria}. In:
	Proceedings of the Fortieth Annual ACM Symposium on Theory of Computing (STOC 2008). pp.
	355--364. ACM, New York, NY, USA (2008)
	
	\bibitem{VPBJ09}
	Voice, T., Polukarov, M., Byde, A., Jennings, N.R.: On the Impact of Strategy
	and Utility Structures on Congestion-Averse Games. In:
	% Leonardi, S. (ed.)
	Proceedings of the 5th International Workshop on Internet and Network Economics
	(WINE 2009). pp. 600--607. Springer Berlin Heidelberg (2009)
	
\end{thebibliography}

\newpage
\appendix
\section*{Appendix}

\section{Omitted Proofs}

\subsection{Proof of Proposition~\ref{p-equal}}
\begin{proof}
	Let $\Gamma$ be a congestion game in which all players use the aggregation function $L_p$ for some $p\in \mathcal{R}$. Then $\Gamma$ can be replaced by the equivalent standard congestion game $\Gamma'$ with cost functions $c_r'(S)=c_r(S)^p$. The cost of player $i$ in state $S$ in $\Gamma$ is equal to $c_i(S)=c_i'(S)^{\frac{1}{p}}$. Since it is a strictly increasing function, taking the $p$-th root is strictly monotone, and thus every best-response strategy in $\Gamma'$ is also a best-response strategy in $\Gamma$. Since pure Nash equilibria are guaranteed to exist in standard congestion games~\cite{RO73}, we can conclude that $\Gamma'$ possesses a PNE, which implies that $\Gamma$ possesses a PNE as well.
	\qed
\end{proof}

\subsection{Proof of Theorem~\ref{t-notexist}}
\begin{proof}
	We prove the statement by constructing for any arbitrary $p=p_1$ and $q=p_2$ a game as described in the theorem. Let $p,q\in \mathbb{N}$ with $q> p$, and $z:=\frac{q+p}{2}$. Consider the following game:
	
	$\Gamma=(N,R,\left(\Sigma_i\right)_{i\in N}, \left(c_r\right)_{r\in R},\left(p_i\right)_{i\in N})$ with $N=\{1,2\}$, $R=\{r_1,\dots,r_6\}$, $\Sigma_1=\left\{\{r_1,r_3,r_5\},\{r_2,r_4,r_6\}\right\}$, $\Sigma_2=\left\{\{r_1,r_3,r_6\},\{r_2,r_4,r_5\}\right\}$, $c_{r_j}(1)=0$ for $1\leq j\leq 6$, $c_{r_j}(2)=1$ for $1\leq j\leq 4$, and $c_{r_j}(2)=2^\frac{1}{z}$ for $j\in \{5,6\}$,	$p_1=p$ and $p_2=q$.
	
	The strategy spaces are constructed in such a way that in every state both players share either two resources ($r_1$ and $r_3$ or $r_2$ and $r_4$) or exactly one resource ($r_5$ or $r_6$). The resources $r_5$ and $r_6$ are more expensive; since player one has the lower $p_i$-value, she will be more willing to share the more expensive resource, while the player with the higher $p_i$-value will be more willing to share a higher number of resources. We will now analyze the costs of the players in the different states, disregarding the $p_i$-th root that is taken in the computation of the $L_p$-norm, since it does not influence the pure Nash equilibria.
	
	Let $S_1$ be a state in which the players share exactly one resource, and $S_2$ a state in which they share two resources. Since $p<z$ and $q>z$, we get:
	\begin{align*}
		c_1(S_1) =\left(2^\frac{1}{z}\right)^p=2^\frac{p}{z}<2=c_1(S_2)\text{ and }
		c_2(S_1) =\left(2^\frac{1}{z}\right)^q=2^\frac{q}{z}>2=c_2(S_2)
	\end{align*}
	
	Hence, player 1 prefers states in which only one resource is shared, whereas player 2 prefers states in which two cheaper resources are shared.
	Since both players can deviate from a state equivalent to $S_1$ to an $S_2$-state, and vice versa, none of the states of this game is a pure Nash equilibrium.
	\qed
\end{proof}

\subsection{Proof of Theorem~\ref{PoA}}
\begin{proof}
	For the proof, we use the smoothness result by Christodolou et al.~\cite{CH11} for an arbitrary $z\in \mathbb{N}$. Furthermore, we apply the fact that the cost of any player in any state is at most the cost she would incur according to the $L_1$-norm, and the fact that the cost according to the $L_p$-norm is at least $d^{\frac{1}{p}-1}$ times the $L_1$-cost.
	
	For a state $S=(S_1,\dots,S_n)$ and a player $i$ we denote by $c_i^1(S)=\sum_{r\in S_i}c_r(S)$ the cost of player $i$ according to the $L_1$-norm and we denote by $cost^1(S)=\sum_{i\in N}c_i^1(S)$ the total cost of all players in state $S$ according to the $L_1$-norm. For any two states $S=(S_1,\dots,S_n)$ and $S'=(S_1',\dots,S_n')$ and any $z\in \mathbb{N}$, we get:
	\begin{align*}
	\sum_{i\in N}c_i(S_i',S_{-i})&\leq \sum_{i\in N}c_i^1(S_i',S_{-i})
	\leq \frac{z^2+3z+1}{2z+1}\cdot cost^1(S')+\frac{1}{2z+1}\cdot cost^1(S)\\
	&\leq d^{1-\frac{1}{p}}\cdot\frac{z^2+3z+1}{2z+1}\cdot cost(S')+\frac{d^{1-\frac{1}{p}}}{2z+1}\cdot cost(S),
	\end{align*}
	which proves that $\Gamma$ is $\left(d^{1-\frac{1}{p}}\cdot\frac{z^2+3z+1}{2z+1},d^{1-\frac{1}{p}}\cdot\frac{1}{2z+1}\right)$-smooth. 
	Since every $(\lambda,\mu)$-smooth game has a price of anarchy of at most $\frac{\lambda}{1-\mu}$, this implies that the price of anarchy of $\Gamma$ is bounded by 
	$d^{1-\frac{1}{p}}\cdot\frac{z^2+3z+1}{2z+1-d^{1-\frac{1}{p}}},$
	for every $z\in \mathbb{N}$ with $2z+1>d^{1-\frac{1}{p}}$. In particular, this holds for the value of $z=\left\lfloor \frac{1}{2}\cdot\left(d^{1-\frac{1}{p}}+\sqrt{5+6\cdot(d^{1-\frac{1}{p}}-1)}\right.\right.$ $\left.\left.\overline{+(d^{1-\frac{1}{p}}-1)^2}\right)\right\rfloor$ given in the theorem, which according to~\cite{CH11} is optimal with respect to the achieved bound on the price of anarchy.
	\qed
\end{proof}

\subsection{Proof of Theorem~\ref{Lp-approx}}
\begin{proof}
	Let $z:=\left(\frac{1}{2}\cdot \left(\frac{1}{p}+\frac{1}{q}\right)\right)^{-1}$.
	We prove the statement by showing that the function $\Phi$ defined as
	$\Phi(S)=\sum_{r\in R}\sum_{i=1}^{n_r(S)}c_r(i)^z$
	is a $\beta$-approximate potential function: i.e., whenever a player improves her personal cost by more than a factor of $\beta$, the value of $\Phi$ is bound to decrease. This definition is a slight variation of Rosenthal's potential function~\cite{RO73}. It has the property that its value decreases whenever a player performs a change in strategy which would be beneficial if her costs were aggregated according to the $L_z$-norm.
	
	As we will show, a player can improve her cost in any state by at most a factor of $\beta$ without decreasing her cost according to the $L_z$-norm. Hence, if a player improves by more than a factor of $\beta$, she also improves her cost according to the $L_z$-norm. Due to the definition of $\Phi$, this implies that the value of $\Phi$ decreases, which yields the theorem.
	
	Until the last step of the proof, we disregard the fact that players take the $p_i$-th root of their cumulated costs (and pretend that the cost of player $i$ in state $S=(S_i,S_{-i})$ was $\sum_{r\in S_i}c_r(S)^{p_i}$).
	We denote by $c_i^z(S_i,S_{-i})=\sum_{r\in S_i}c_r(S_i,S_{-i})^z$ the cost of player $i$ in state $(S_i,S_{-i})$ if it was computed according to the $L_z$-norm (i.e. $p_i=z$).
	Consider a player $i$ with aggregation value $p_i$ in state $S=(S_i,S_{-i})$. Assume that $i$ can improve her cost by deviating to a state $S'=(S_i',S_{-i})$ without improving her costs according to the $L_z$-norm, i.e. $c_i(S')<c_i(S)$, but $c_i^z(S')\geq c_i^z(S)$.
	
	We now distinguish between the cases $p_i\leq z$ and $p_i\geq z$.
	
	For the case of $p_i\geq z$, we get:
	\begin{align*}
		c_i(S_i',S_{-i})&=\sum_{r\in S_i'}c_r(S_i',S_{-i})^{p_i}=\sum_{r\in S_i'}\left(c_r(S_i',S_{-i})^z\right)^\frac{p_i}{z}\\
		&\geq \sum_{r\in S'} \left(\frac{1}{|S_i'|}\cdot c_i^z(S_i',S_{-i})\right)^\frac{p_i}{z}\geq|S_i'|\cdot \left(\frac{1}{|S_i'|}\cdot c^z_i(S)\right)^\frac{p_i}{z}.
	\end{align*}
	The first inequality holds since the value of $\sum_{i=1}^k x_i^l$ for fixed $x_1+\dots+x_k$ and $l\geq 1$ is minimized if $x_i=x_j$ holds for all $i$ and $j$. The second inequality is correct since by assumption it holds that $c_i^z(S_i',S_{-i})\geq c_i^z(S)$.
	
	On the other hand, we have 
	\begin{align*}
		c_i(S)&=\sum_{r\in S_i}c_r(S)^{p_i}=\sum_{r\in S_i}\left(c_r(S)^z\right)^\frac{p_i}{z} \leq \left(\sum_{r\in S_i}c_r(S)^z\right)^\frac{p_i}{z}=c_i^z(S)^\frac{p_i}{z}.
	\end{align*}
	Putting things together, we get
	\begin{align*}
		\frac{c_i(S)}{c_i(S_i',S_{-i})}&\leq\frac{|S'_i|^\frac{p_i}{z}\cdot c_i^z(S)^\frac{p_i}{z}}{|S_i'|\cdot c_i^z(S)^\frac{p_i}{z}}\leq \frac{d^\frac{p_i}{z}}{d}=d^{\frac{p_i}{z}-1}.
	\end{align*}
	We now regard the case $p_i\leq z$, which implies $\frac{p_i}{z}\leq 1$.
	
	Analogously to the first case, we use that for fixed $x_1+\dots+x_k$ and $l\leq 1$ the value of $\sum_{i=1}^k x_i^l$ is minimized if $x_i=x_j$ holds for all $i$ and $j$ and maximized if all but one summand are 0. This is a direct reversal of the properties for $l\geq 1$.
	
	Thus, following the same line of argument as in the first case, we get
	\begin{align*}
		c_i(S_i',S_{-i})&\geq c_i^z(S)^\frac{p_i}{z}\text{ and}\\
		c_i(S)& \leq |S_i|\cdot \left(\frac{1}{|S_i|}\cdot c^z_i(S)\right)^\frac{p_i}{z}\leq d\cdot \left(\frac{1}{d}\cdot c^z_i(S)\right)^\frac{p_i}{z},\text{ which implies}\\
		\frac{c_i(S)}{c_i(S_i',S_{-i})}&\leq \frac{d}{d^\frac{p_i}{z}}=d^{1-\frac{p_i}{z}}
	\end{align*}
	Now we will take into account that players take the $p_i$-th root of their cumulated costs. This makes a big difference for approximate equilibria, since every $\gamma$-improvement step in the game without taking roots corresponds to a $\gamma^\frac{1}{p_i}$-improvement step in the original game. We denote by $c_i^o(S)$ the cost of player $i$ in the original game. Using the definition of $z$, we analyze this for both cases separately, starting with $p_i\geq z$.
	\begin{align*}
		\frac{c_i^o(S)}{c_i^o(S_i',S_{-i})}&\leq\left(d^{\frac{p_i}{z}-1}\right)^\frac{1}{p_i}=d^{\frac{1}{z}-\frac{1}{p_i}}\leq d^{\frac{1}{z}-\frac{1}{q}}=d^{\frac{1}{2}\cdot\left(\frac{1}{p}+\frac{1}{q}\right)-\frac{1}{q}}=d^{\frac{1}{2}\left(\frac{1}{p}-\frac{1}{q}\right)}
	\end{align*}
	Analogously, we get for the case of $p_i\leq z$:
	\begin{align*}
		\frac{c_i^o(S)}{c_i^o(S_i',S_{-i})}&\leq\left(d^{1-\frac{p_i}{z}}\right)^\frac{1}{p_i}=d^{\frac{1}{p_i}-\frac{1}{z}}\leq d^{\frac{1}{p}-\frac{1}{z}}=d^{\frac{1}{p}-\frac{1}{2}\cdot\left(\frac{1}{p}+\frac{1}{q}\right)}=d^{\frac{1}{2}\left(\frac{1}{p}-\frac{1}{q}\right)}
	\end{align*}
	This shows that no player can improve her cost by more than a factor of $\beta=d^{\frac{1}{2}\left(\frac{1}{p}-\frac{1}{q}\right)}$ without decreasing her cost with respect to the $L_z$-norm. Hence, every $\beta$-improvement step of any player decreases the value of $\Phi$. This implies that sequences of $\beta$-improvement steps can not contain cycles and therefore reach a $\beta$-approximate equilibrium after a finite number of steps.
	\qed
\end{proof}

\subsection{Proof of Theorem~\ref{Lpnonapprox}}
\begin{proof}
	\begin{figure}[b]
		
	\end{figure}
	For any given $p<q$, we define $z:=\left(\frac{1}{2}\cdot \left(\frac{1}{p}+\frac{1}{q}\right)\right)^{-1}$. Based on this, we construct the following 2-player congestion game $\Gamma$ with $N=\{1,2\}$, $R=\{r_1^1,r_1^2,r_2^1,r_2^2,\dots,r_d^1,r_d^2\}$, $\Sigma_1=\left\{\{r_1^1,r_2^1,\dots, r_{d-1}^1,r_d^1\},\{r_1^2,r_2^2,\dots, r_{d-1}^2,r_d^2\}\right\}$,\\	$\Sigma_2=\left\{\{r_1^1,r_2^1,\dots, r_{d-1}^1,r_d^2\},\{r_1^2,r_2^2,\dots, r_{d-1}^2,r_d^1\}\right\}$, $c_{r_j^k}(0)=0$ and $c_{r_j^k}(1)=1$ for $1\leq j\leq d-1$ and $k\in\{1,2\}$, $c_{r_d^k}(0)=1$ and $c_{r_d^k}(1)=d^\frac{1}{z}$ for $k\in\{1,2\}$, 	$p_1=p$ and $p_2=q$.
	
	We first remark that $d$ is indeed the size of the strategy containing most resources (actually, all strategies have cardinality $d$). Due to the construction of the strategy sets, in every state both players either share the $d-1$ resources $r_1^j,\dots, r_{d-1}^j$ or only the resource $r_d^j$ for a $j\in \{1,2\}$.
	
	In the former case, player 1 incurs a cost of $d^\frac{1}{p}$ and player 2 a cost of $d^\frac{1}{q}$. In the latter case, both players incur a cost of $d^\frac{1}{z}$, which is better for player 1, but worse for player 2. Hence, when deviating to the more preferable state, player 1 improves her cost by a factor of $\frac{d^\frac{1}{p}}{d^\frac{1}{z}}=d^{\frac{1}{p}-\frac{1}{z}}=d^{\frac{1}{p}-\frac{1}{2}\cdot\left(\frac{1}{p}+\frac{1}{q}\right)}=d^{\frac{1}{2}\left(\frac{1}{p}-\frac{1}{q}\right)}.$
	Likewise, in the opposite direction player 2 improves her cost by a factor of
	$d^{\frac{1}{z}-\frac{1}{q}}=d^{\frac{1}{2}\cdot\left(\frac{1}{p}+\frac{1}{q}\right)-\frac{1}{q}}=d^{\frac{1}{2}\left(\frac{1}{p}-\frac{1}{q}\right)}.$
	
	Starting in an arbitrary state, both players will alternate between the two types of states. In every improvement step the respective player improves her cost by a factor of $d^{\frac{1}{2}\left(\frac{1}{p}-\frac{1}{q}\right)}$. This shows that $\Gamma$ contains no $\beta$-approximate pure Nash equilibrium for any $\beta<d^{\frac{1}{2}\left(\frac{1}{p}-\frac{1}{q}\right)}$.
	\qed
\end{proof}

\subsection{Proof of Theorem~\ref{Lpapproxcomp}}
\begin{proof}
	We prove the statement for $d\geq 3$. Afterwards, we will shortly point out which modifications are necessary for the case $d=2$.
	We reduce from the independent set problem. Since the size of the players' strategies may not exceed $d$, we use the independent set problem on graphs with bounded node degrees ($IS_b$). It is NP-complete for degrees of at least 3~\cite{GJ02}. Since $d\geq 3$, we can reduce from $IS_d$, which ensures $|E_v|\leq d$ for all $v\in V$.   
	Let $z=\left(\frac{1}{2}\cdot\left(\frac{1}{p}+\frac{1}{q}\right)\right)^{-1}$. For any instance $\left<G=(V,E),k\right>$, we construct a congestion game with $L_p$-aggregation functions $\Gamma$ as follows:
	
		\begin{center}
			$\Gamma=(N,R,\left(\Sigma_i\right)_{i\in N},
			\left(c_r\right)_{r\in R},\left(p_i\right)_{i\in N})$:\\
		\end{center}
		$N=\{1,\dots,k,c,k+1,k+2\}$,\\
		$R=\{r_e~|~e\in E\}\cup \{r_1^1,r_1^2,\dots,r_d^1,r_d^2,r_c\}$,\\
		$\Sigma_i=\{\{r_e~|~e\in E_v\}~|~v\in V\} \cup \{\{r_c\}\}$ for $1\leq i \leq k$,\\
		$\Sigma_c=\{\{r_c\},\{r_d^1,r_d^2\}\}$,\\
		$\Sigma_{k+1}=\{\{r_1^1,\dots,r_{d-1}^1,r_d^1\},\{r_1^2,\dots,r_{d-1}^2,r_d^2\}\}$,\\
		$\Sigma_{k+2}=\{\{r_1^1,\dots,r_{d-1}^1,r_d^2\},\{r_1^2,\dots,r_{d-1}^2,r_d^1\}\}$,\\
		$c_{r_e}(1)=0$ and $c_{r_e}(x)=2\cdot d^3$ for $x\geq 1$ and for all $e\in E$,\\
		$c_{r_c}(1)=0$ and $c_{r_c}(x)=2\cdot d^2$ for $x\geq 1$,\\
		$c_{r_l^j}=(0,1)$ for $1\leq l\leq d-1$ and $j\in \{1,2\}$,\\
		$c_{r_d^j}=(0,1,d^\frac{1}{z})$ for $j\in \{1,2\}$,\\
		$p_{k+2}=q$ and $p_i=p$ for all other players.

	If the connection player $c$ does not interfere in the sub game defined by the players $k+1$ and $k+2$, the state $\left(\{r_1^1,\dots,r_d^1\},\{r_1^2,\dots,r_{d-1}^2,r_d^1\}\right)$ is a PNE, since both players have a cost of 1, which can not be improved. On the other hand, if $c$ allocates $r_d^1$ and $r_d^2$, the game is equivalent to the game used in the proof of Theorem \ref{Lpnonapprox}, which does not possess a $\beta$-approximate PNE for any $\beta<d^{\frac{1}{2}\cdot\left(\frac{1}{p}-\frac{1}{q}\right)}$.
	
	Obviously, if the $k$ node players allocate disjoint edge sets, all of them and the connection player have a cost of 0 without interfering in the sub game of player $k+1$ and $k+2$. If two players allocate non-disjoint edge sets, they incur a cost of at least $2\cdot d^3$, which can be improved by a factor of at least $d$ by allocating $r_c$. In this case, it is beneficial for player $c$ to join the sub game of player $k+1$ and $k+2$ and annihilate the existence of a PNE in this game. By doing so, player $c$ improves her cost by a factor of at least $\frac{2\cdot d^2}{2\cdot d}=d$.
	
	In summary, $\Gamma$ possesses a state in which no player can improve by a factor of at least $d^{\frac{1}{2}\cdot\left(\frac{1}{p}-\frac{1}{q}\right)}$ if and only if $G$ contains an independent set of size $k$. This concludes the proof for $d>2$. We will now shortly discuss the case $d=2$, which has to be handled separately since the independent set problem is not NP-hard for graphs with node degrees of at most 2.
	
	However, we can reduce from $IS_3$ by introducing auxiliary resources and players. For every node $v$ with three edges $e_1$, $e_2$ and $e_3$, we introduce a node player $v$ and a resource $r_v$. The strategies available to $v$ are given by $\{r_v\}$, $\{r_{e_2},r_{e_3}\}$ and $\{r_c\}$. In the strategy sets of the $k$ original node players, we replace $E_v$ by $\{r_v,r_{e_1}\}$. The cost function of $r_v$ is defined to be equal to the costs of the edge resources.
	
	Given this, we can construct a PNE in $\Gamma$ from an independent set of size $k$ as follows:
	\begin{itemize}
		\item For any node $v$ that is part of the independent set one player chooses her strategy $E_v$ or, if $v$ has degree 3, the strategy $\{r_v,r_{e_1}\}$.
		\item All auxiliary node players for a node $v$ play their strategy $\{e_2,e_3\}$ if $v$ is part of the independent set, and $\{r_v\}$ otherwise.
	\end{itemize}
	In the described state all node players incur a cost of 0 and none of them allocates $r_c$. Thus, a PNE is established if player $c$ allocates $r_c$ and does not interfere in the sub game containing the players $k+1$ and $k+2$.
	
	On the other hand, if $G$ does not contain an independent set of size $k$, it holds in every state that two node players allocate the resource $r_v$ for the same node $v$ or the resource $r_e$ for the same edge $e$, or at least one node player allocates $r_c$. As described for $d\geq 3$, this implies that in every state at least one player can improve her cost by a factor of at least $d^{\frac{1}{2}\cdot\left(\frac{1}{p}-\frac{1}{q}\right)}$, which completes the proof for all $d\geq 2$.
	\qed
\end{proof} 
\subsection{Proof of Theorem~\ref{strongnonmonotone}}
\begin{proof}
	We proof the statement by using the vectors in the definition of strongly non-monotone functions in order to construct a game that is equivalent to the well-known ``matching pennies'' game (cf. e.g.~\cite{KL12}). Let $x$, $y$, $z$, $j$, and $k$ be defined as in the definition. We define the following game with two players:
	
		$N=\{1,2\}$, $R=\{r_1^1,\dots,r_{d-1}^1,r_1^2,\dots,r_{d-1}^2,r_h,r_t\}$,\\
		$\Sigma_1=\left\{\{r_1^1,\dots,r_{d-1}^1,r_h\},\{r_1^1,\dots,r_{d-1}^1,r_t\}\right\}$,\\
		$\Sigma_2=\left\{\{r_1^2,\dots,r_{d-1}^2,r_h\},\{r_1^2,\dots,r_{d-1}^2,r_t\}\right\}$,\\
		$c(r_i^1)=(x_i)$ for $i<j$ and $c(r_i^1)=(x_{i+1})$ for $i\geq j$,\\
		$c(r_i^2)=(z_i)$ for $1\leq i\leq d-1$,\\
		$c(r_h)=c(r_t)=(x_j,y_k)$.
		
	Clearly, both $\Sigma_1$ and $\Sigma_2$ form the sets of the bases of a matroid.
	For the first player, it is preferable to use $r_h$ if the second player uses $r_h$, and $r_t$ if player 2 uses $r_t$, since 
	$f(x_1,\dots,x_{j-1},x_{j+1},\dots,x_k,y_k,x_{k+1},\dots,x_d)<f(x_1,\dots,x_d)$. 
	On the other hand, player 2 prefers to allocate the resource different from player 1, since
	$f(z_1,\dots,x_j,\dots,z_{d-1})<f(z_1,\dots,y_k,\dots,z_{d-1})$. Hence, none of the four states of this game is a pure Nash equilibrium.\qed
\end{proof}

\section{Examples Values for the Approximation Factors}
\label{examplevalues}
Example values for the two approaches for linear cost functions, based on the simplifications $O(1)=1$ and $\epsilon=0$:
	\begin{center}
		\begin{tabular}{|l| l|| c|c ||c|c|}
			\hline
			p&q&\multicolumn{2}{|c||}{1st approach}&\multicolumn{2}{|c|}{2nd approach}\\
			
			&&d=2&d=10&d=2&d=10\\
			1&2&2.83&6.32&1.59&2.37\\
			\hline
			2&3&3.17&9.28&2.54&2.91\\
			\hline
			1&$\infty$&4&20&2.83&6.32\\
			\hline
			2&$\infty$&4&20&4.76&7.11\\
			\hline
			$10$&$\infty$&4&20&20.71&22.44\\
			\hline
		\end{tabular}
	\end{center}

\end{document}